\documentclass[11pt]{article}
\usepackage{amsfonts,amssymb,amsmath,latexsym,ae,aecompl}
\usepackage{graphics}
\usepackage{latexsym}
\usepackage{algorithm}
\usepackage{algorithmic}
\usepackage{tikz}
\usepackage{caption}
\usepackage{subcaption}
\usepackage{subfig}
\newcommand{\qed}{\hfill$\Box$}

\newcommand{\cG}{{\cal G}}
\newcommand{\cB}{{\cal A}}
\newcommand{\cH}{{\cal H}}
\newcommand{\cA}{{\tt Ring}}

\newenvironment{proof}{\noindent {\bf Proof.}}{\qed}

\newtheorem{theorem}{Theorem}[section]
\newtheorem{lemma}{Lemma}[section]
\newtheorem{corollary}{Corollary}[section]

\newtheorem{proposition}{Proposition}[section]

\begin{document}

\baselineskip 0.2in
\parskip      0.1in
\parindent    0em

\bibliographystyle{plain}

\title{{\bf Deterministic Graph Exploration with Advice}}

\author{
Barun Gorain \footnotemark[1]
\and
Andrzej Pelc \footnotemark[1] \footnotemark[2]
}

\date{ }
\maketitle
\def\thefootnote{\fnsymbol{footnote}}

\footnotetext[1]{
\noindent
 D\'{e}partement d'informatique, Universit\'{e} du Qu\'{e}bec en Outaouais,
Gatineau, Qu\'{e}bec J8X 3X7,
 Canada. E-mails:
{\tt baruniitg123@gmail.com}, {\tt pelc@uqo.ca}}

\footnotetext[2]{
\noindent
Research supported in part by NSERC  Discovery Grant 8136 -- 2013  and by the
Research Chair in Distributed Computing of the
Universit\'{e} du Qu\'{e}bec en Outaouais.
}

\begin{abstract}
We consider the fundamental task of graph exploration. An $n$-node graph has unlabeled nodes,
and all ports at any node of degree $d$ are arbitrarily numbered $0,\dots, d-1$. A mobile agent, initially
situated at some starting node $v$, has to visit all nodes and stop.
The {\em time} of the exploration is the number of edge traversals. We consider the problem of how much knowledge
the agent has to have a priori, in order to explore the graph in a given time, using a deterministic algorithm.  Following the paradigm of {\em algorithms with advice},
this a priori information (advice) is provided to the agent by an {\em oracle}, in the form of a binary string, whose length is called the {\em size of advice}.
We consider two types of oracles. The {\em instance oracle} knows the entire instance of the exploration problem,
i.e., the port-numbered map of the graph and the starting node of the agent in this map. The {\em map oracle} knows the
port-numbered map of the graph but does not know the starting node of the agent. What is the minimum size of advice that must be given to the agent
by each of these oracles, so that the agent explores the graph in a given time?

We first consider exploration in polynomial time, and determine the exact minimum size of advice to achieve it. Indeed, we prove that some advice
of size $\log\log\log n -c$, for any constant $c$, is sufficient to permit polynomial exploration, and that no advice of size  $\log\log\log n -\phi(n)$,
where $\phi$ is any function diverging to infinity, can help to do this. These results hold both for the instance and for the map oracles.

On the other side of the spectrum, when advice is large, there are two natural time thresholds:  $\Theta(n^2)$ for a map oracle, and $\Theta(n)$ for an instance oracle.
This is because, in both cases, these time benchmarks can be achieved  with sufficiently large advice (advice of size $O(n\log n)$ suffices).
We show that, with a map oracle, time $\Theta(n^2)$ cannot be improved in general, regardless of the size of advice. What is then the smallest advice to achieve time $\Theta(n^2)$ with a map oracle?
We show that this smallest size of advice is larger than $n^\delta$, for any $\delta <1/3$.

For large advice, the situation changes significantly when we allow an instance oracle instead of a map oracle. In this case, advice of size $O(n\log n)$ is enough to  achieve time $O(n)$.
Is such a large advice needed to achieve linear time? We answer this question affirmatively. Indeed, we show more: with any advice of size $o(n\log n)$, the time of exploration must be at least $n^\epsilon$, for any
$\epsilon <2$, and with any advice of size $O(n)$, the time must be  $\Omega(n^2)$.

We finally look at hamiltonian graphs, as for them it is possible to achieve the absolutely optimal exploration time $n-1$, when sufficiently large advice (of size  $O(n\log n)$) is given by an instance oracle.
We show that a map oracle cannot achieve this: regardless of the size of advice,  the time of exploration must be $\Omega(n^2)$, for some hamiltonian graphs. On the other hand, even for the instance oracle,
with advice of size $o(n\log n)$, optimal time $n-1$ cannot be achieved: indeed, we show that the time of exploration with such advice must sometimes exceed the optimal time $n-1$ by a summand
$n^\epsilon$, for any $\epsilon <1$.


{\bf Keywords:} algorithm, graph, exploration, mobile agent, advice

\end{abstract}

\pagebreak

\section{Introduction}

Exploration of networks by visiting all of their nodes is one of the basic tasks performed by a mobile agent in networks.
In applications, a software agent may need to collect data placed at nodes of a network, or a mobile robot
may need to collect samples of air or ground in a contaminated mine whose corridors form links of a network, with
corridor crossings represented by  nodes.

The network is modeled as a simple connected
undirected graph $G=(V,E)$ with $n$ nodes, called {\em graph} in the sequel.
Nodes are unlabeled, and all ports at any node of degree $d$ are arbitrarily numbered $0,\dots, d-1$.
The agent is initially situated at a starting node $v$ of the graph. When the agent located at a current node $u$ gets to a neighbor $w$ of $u$
by taking port $i$, it learns the port $j$ by which it enters node $w$ and it learns the degree of $w$.
The agent has to visit all nodes of the graph and stop.

The {\em time} of the exploration is the number of edge traversals. We consider the problem of how much knowledge
the agent has to have a priori, in order to explore the graph in a given time, using a deterministic algorithm. It is well-known that some information is necessary, as witnessed even by the class of rings in which
ports at all nodes are numbered 0,1 in clockwise order. Navigating in such a ring, the agent cannot learn its size. If there existed an exploration algorithm not using any a priori knowledge,
then it would have to stop after some $t$ steps in every ring, and hence would fail to explore a $(t+2)$-node ring.

Following the paradigm of {\em algorithms with advice} \cite{AKM01,DP,EFKR,FGIP,FIP1,FIP2,FKL,FP,FPR,GPPR02,IKP,KKP05,SN},
this a priori information (advice), needed for exploration, is provided to the agent by an {\em oracle}, in the form of a binary string, whose length is called the {\em size of advice}.
We consider two types of oracles. An {\em instance oracle} knows the entire instance of the exploration problem,
i.e., the port-numbered map of the graph and the starting node of the agent in this map. A {\em map oracle} knows the
port-numbered map of the graph but does not know the starting node of the agent. Formally, a map oracle is a function $f:{ \cal G} \longrightarrow S$, where $\cal G$ is the set of graphs and $S$ is the set of finite binary strings. An instance oracle is a function $f:{ \cal I} \longrightarrow S$, where $\cal I$ is the set of couples $(G,v)$,
with $G\in \cal G$ and $v$ being the starting node of the agent in graph $G$.
The advice $s$ is an input to an exploration algorithm.  We say that exploration in time $t$ with advice of size $x$ given by an instance oracle is possible, if there exists advice of length $x$ depending on the instance $(G,v)$,
 and an exploration algorithm using this advice, which explores every graph in time $t$, starting from node $v$. Likewise, we say that
 exploration in time $t$ with advice of size $x$ given by a map oracle is possible,  if there exists advice of length $x$ depending on the graph $G$,
 and an exploration algorithm using this advice, which explores every graph in time $t$, starting from any node.
(Integers $x$ and $t$ depend on the size of the graph.) Proving that such an exploration is possible consists in showing an oracle of the appropriate type giving advice of size $x$, and an exploration algorithm using this advice and working in time $t$, for any graph and any starting node.
Proving that such an exploration is impossible consists in showing that, for any oracle of the appropriate type giving advice of size $x$, and for any exploration algorithm
using it, there exists a graph and a starting node for which this algorithm will exceed time $t$.

The central question studied in this paper is:
\begin{quotation}
What is the minimum size of advice that has to be given to the agent
by an instance oracle (resp. by a map oracle) to permit the agent to explore any graph in a given time?
\end{quotation}

Our main contribution are negative results of two types:
\begin{itemize}
\item
impossibility results showing that the less powerful map oracle cannot help to achieve the same exploration time as the more powerful instance oracle, regardless of the size of advice;
\item
lower bounds showing that the size of some natural advice leading to a simple exploration in a given time cannot be improved significantly.
\end{itemize}
While in most cases our bounds on the size of advice are asymptotically tight,  in one case the remaining gap  is cubic.

\subsection{Our results}

We first consider exploration in polynomial time, and determine the exact minimum size of advice to achieve it. Indeed, we prove that some advice
of size $\log\log\log n -c$, for any constant $c$, is sufficient to permit polynomial exploration of all $n$-node graphs, and that no advice of size  $\log\log\log n -\phi(n)$,
where $\phi$ is any function diverging to infinity, can help to do this. Both these results hold both for the instance and for the map oracles.

On the other side of the spectrum, when advice is large, there are two natural time thresholds:  $\Theta(n^2)$ for a map oracle, and $\Theta(n)$ for an instance oracle.
This is because, in both cases, these time benchmarks can be achieved  with sufficiently large advice (advice of size $O(n\log n)$ suffices).
We show that, with a map oracle, time $\Theta(n^2)$ cannot be improved in general, regardless of the size of advice. What is then the smallest advice to achieve time $\Theta(n^2)$ with a map oracle?
We show that this smallest size of advice is larger than $n^\delta$, for any $\delta <1/3$.

For large advice, the situation changes significantly when we allow an instance oracle instead of a map oracle. In this case, advice of size $O(n\log n)$ is enough to  achieve time $O(n)$.
Is such a large advice needed to achieve linear time? We answer this question affirmatively. Indeed, we show more: with any advice of size $o(n\log n)$, the time of exploration must be at least $n^\epsilon$, for any
$\epsilon <2$, and with any advice of size $O(n)$, the time must be  $\Omega(n^2)$.

We finally look at hamiltonian graphs, as for them it is possible to achieve the absolutely optimal exploration time $n-1$, when sufficiently large advice (of size  $O(n\log n)$) is given by an instance oracle.
We show that a map oracle cannot achieve this: regardless of the size of advice,  the time of exploration must be $\Omega(n^2)$, for some hamiltonian graphs. On the other hand, even for an instance oracle,
with advice of size $o(n\log n)$, optimal time $n-1$ cannot be achieved: indeed, we show that the time of exploration with such advice must sometimes exceed the optimal time $n-1$ by a summand
$n^\epsilon$, for any $\epsilon <1$.

Our results permit to compare advice of different size and of different quality. The size is defined formally, and for quality we may say that advice given by an instance oracle is superior to advice given by a map oracle, because an instance oracle, seeing not only the graph but also the starting node of the agent, can use the allowed bits of advice in a better way. Looking from this perspective it turns out that both size and quality of advice provably matter. The fact that quality of advice matters is proved
by the following pair of results: for a map oracle, time $\Theta(n^2)$ cannot be beaten, regardless of the size of advice, while for an instance oracle time $O(n)$
can be achieved with $O(n\log n)$ bits of advice. The fact that the size of advice matters (with the same quality) is proved by the following pair of results: for an instance oracle, time $O(n)$ can be achieved with $O(n\log n)$ bits of advice, but with $o(n\log n)$ bits of advice time must be at least $n^\epsilon$, for any
$\epsilon <2$.

\subsection{Related work}

The problem of
exploration and navigation of mobile agents in an unknown environment
has been extensively studied in the literature for many decades (cf. the survey \cite{RKSI}).
The explored environment has been modeled in two distinct ways:
either as a geometric terrain in the plane, e.g., an
unknown terrain with convex obstacles \cite{BRS},
or a room with polygonal \cite{DKP} or rectangular \cite{BBFY} obstacles, or
as we do in this paper, i.e.,
as a graph, assuming that the agent may only move along its edges. The graph
model can be further specified in two different ways:
either the graph is directed, in which case the agent can move only from
tail to head of a directed edge  \cite{AH,BFRSV,BS,DP}, or the graph is undirected (as we assume)
and the agent can traverse edges in both directions  \cite{ABRS,BRS2,DKK,PaPe,PaPe2}.
The efficiency measure adopted in most papers dealing with exploration of graphs is the time (or cost)
of completing this task, measured by the number of edge traversals by the agent.
Some authors impose further restrictions on the moves of the agent.
It is assumed that the agent has either a restricted tank \cite{ABRS,BRS2},
and thus has to periodically return to the base for refueling, or that it is attached to the
base by a rope or cable of restricted length \cite{DKK}.

Another direction of research concerns
exploration of anonymous graphs.
In this case it is impossible
to explore arbitrary graphs and stop after exploration, if no marking of nodes is allowed, and if nothing is known about the graph.
Hence some authors \cite{BFRSV,BS}
allow {\em pebbles} which the agent can drop on nodes to recognize already visited ones, and
then remove them and drop them in other places. A more restrictive scenario assumes
a stationary token that is fixed at the starting node of the agent \cite{CDK,PeTi}.
Exploring
anonymous graphs without the possibility of marking nodes (and thus possibly without stopping)
is investigated, e.g., in \cite{DFKP,FI}.
The authors
concentrate attention not on the cost of exploration but on the minimum amount of memory sufficient
to carry out this task. In the absence of marking nodes,  in order to guarantee stopping after exploration, some knowledge about the graph is required,
e.g., an upper bound on its size \cite{CDK,Re}.

Providing nodes or agents with arbitrary kinds of information that can be used to perform network tasks more efficiently has been previously
proposed in \cite{AKM01,DP,EFKR,FGIP,FIP1,FIP2,FKL,FP,FPR,GPPR02,IKP,KKP05,SN} in contexts ranging from graph coloring to broadcasting and leader election. This approach was referred to as
{\em algorithms with advice}.
The advice is given either to nodes of the network or to mobile agents performing some network task.
In the first case, instead of advice, the term {\em informative labeling schemes} is sometimes used, if different nodes can get different information.

Several authors studied the minimum size of advice required to solve
network problems in an efficient way.
 In \cite{KKP05}, given a distributed representation of a solution for a problem,
the authors investigated the number of bits of communication needed to verify the legality of the represented solution.
In \cite{FIP1}, the authors compared the minimum size of advice required to
solve two information dissemination problems using a linear number of messages.
In \cite{FKL}, it was shown that advice of constant size given to the nodes enables the distributed construction of a minimum
spanning tree in logarithmic time.
In \cite{DKM,EFKR}, the advice paradigm was used for online problems. In particular, in \cite{DKM} the authors studied online graph exploration with advice in
labeled weighted graphs.
In \cite{FGIP}, the authors established lower bounds on the size of advice
needed to beat time $\Theta(\log^*n)$
for 3-coloring cycles and to achieve time $\Theta(\log^*n)$ for 3-coloring unoriented trees.
In the case of \cite{SN}, the issue was not efficiency but feasibility: it
was shown that $\Theta(n\log n)$ is the minimum size of advice
required to perform monotone connected graph clearing.
In \cite{IKP}, the authors studied radio networks for
which it is possible to perform centralized broadcasting in constant time. They proved that constant time is achievable with
$O(n)$ bits of advice in such networks, while
$o(n)$ bits are not enough. In \cite{FPR}, the authors studied the problem of topology recognition with advice given to nodes.
The topic of \cite{GMP} and \cite{DiPe} was the size of advice needed for fast leader election, resp. in anonymous trees and in arbitrary anonymous graphs.
Exploration with advice was previously studied only for trees \cite{FIP2}, and algorithm performance was measured using the competitive approach.
In the present paper, the performance measure of an algorithm is the order of magnitude of exploration time, and hence the case of trees is trivial, as they can be explored in linear time without any advice.


\section{Exploration in polynomial time}

As a warm-up, we first consider the following question: What is the minimum size of advice permitting the agent to explore any graph in time polynomial in the size of the graph? In this section we give the exact answer to this question, both for the instance oracle and for the map oracle.

It is well-known that, if the agent knows an upper bound $n'$ on the number $n$ of nodes of the graph, then exploration in time polynomial in $n'$ is possible, starting from any node of the graph.
The first result implying this fact was proved in \cite{AKLLR}. The exploration proposed there works in time $O(n'^5\log n')$, and is based on Universal Traversal Sequences (UTS). Later on, an exploration algorithm working in time polynomial in $n'$ based on Universal Exploration Sequences (UXS) was established in \cite{Re}. While the polynomial in the latter paper has much higher degree, the solution from \cite{Re} can be carried out in logarithmic memory. Both UTS and UXS permit to
find a sequence of port numbers to be followed by the agent, regardless of the topology of the graph and of  its starting node. In the case of UTS, the sequence of port numbers to be followed is the UTS itself, and in the case of UXS it is constructed term by term, on the basis of the UXS and of the port number by which the agent entered
the current node. Regardless of which solution is used, we have the following proposition:

\begin{proposition}\cite{AKLLR,Re}\label{uts}
If the agent knows an upper bound $n'$ on the number $n$ of nodes of the graph, there exists an algorithm with input $n'$ that permits the agent starting at any node of the graph to explore the graph and stop after $P(n')$ steps, where $P$ is some polynomial.
\end{proposition}

The positive part of our result on minimum advice is formulated in the following lemma. Its proof is based on Proposition \ref{uts}. The advice given to the agent
is some prefix of the binary representation of  the number $\lfloor \log\log n\rfloor$, on the basis of which the agent computes a rough but sufficiently precise upper bound on the size of the graph which permits it to explore the graph, in time polynomial in its size.

\begin{lemma}\label{poly-ub}
For any positive constant $c$, there exists an exploration algorithm using advice of size $\lfloor \log \log \log n -c \rfloor$,
that works in time polynomial in $n$, for any $n$-node graph.
\end{lemma}
\begin{proof}
Let $A$ be an algorithm  and let $P$ be a polynomial such that, if the agent knows an upper bound $n'$ of the number $n$ of nodes of a graph $G$, then it can explore $G$ in time $P(n')$, starting from any node, using algorithm $A$ with input $n'$.
Without loss of generality suppose $P(m)= m^a$ for some constant $a \in \mathbb{N}$.  We will show that there exists a binary string $s$ of length $\lfloor \log \log \log n -c \rfloor$ such that if $s$ is given to the agent as advice, the agent can explore all nodes of $G$ in time polynomial in $n$. To show the existence of such a string $s$, let $X$ be the binary representation of $\lfloor \log \log n \rfloor$. Let $s$ be the string obtained from $X$ by deleting the last $c+1$ bits of $X$.
The length of $s$ is at most  $\lfloor \log \log \log n -c \rfloor$. This string $s$ is given to the agent as the advice.

Let $s_1$ be string resulting from $s$ by adding $(c+1)$ 1's at the end of $s$. Let $n_1$ be the integer whose binary representation is $s_1$.
Let $N=2^{2^{n_1+1}}$. By definition, we have $n_1 \geq \lfloor \log \log n \rfloor$. Hence, $n_1+1 \geq  \log \log n$, and thus $N \geq n$.

After receiving the string $s$, the agent computes the integer $N$ and performs algorithm $A$ with input $N$.  Since $N \geq n$, the agent correctly explores
the graph in time $P(N)$.

To prove that the exploration time is polynomial in $n$, let $s_0$ be the string which is obtained from $s$ by adding $c+1$ 0's at the end of $s$. Let $n_0$ be the integer whose binary representation is $s_0$.
We have $n_1= n_0+(2^{c+1}-1)$. Therefore, \\
  $$n_0 \le \lfloor \log \log n \rfloor \le n_0+(2^{c+1}-1)$$
  $$n_0 \le  \log \log n  \le n_0+(2^{c+1}-1)+1$$
  $$2^{n_0} \le  \log n  \le 2^{n_0+2^{c+1}}$$
  $$2^{2^{n_0}} \le n \le 2^{2^{n_0+2^{c+1}}}$$
Let $m=2^{2^{n_0}}$. Then $m \le n \le m^{2^{c+1}}$. Also, $P(n) \ge P(m)$.
Hence, the time taken by the agent is at most $P(N)=P(m^{2^{c+1}}) \le m^{a\cdot{2^{c+1}}} = P(m)^{2^{c+1}} \le P(n)^{2^{c+1}}$. Since $c$ is a constant, $P(n)^{2^{c+1}}$ is a polynomial in $n$.
\end{proof}

The next result shows that the upper bound from the previous lemma is tight. Indeed, the following lower bound holds even for oriented rings, i.e., rings in which ports 0 and 1 are in clockwise order at every node.

\begin{lemma}\label{poly-lb}
For any function $\phi:\mathbb{N} \rightarrow \mathbb{N}$ such that $\phi(n) \rightarrow \infty$ as $n \rightarrow \infty$,
 it is not possible to explore an $n$-node oriented ring in polynomial time,
using advice of size at most $ \log \log \log n - \phi(n)$.
\end{lemma}
\begin{proof}
The proof is by contradiction. Suppose that there exists an algorithm $\cA$, permitting the agent to explore an $n$-node oriented ring with at most $\log \log \log n - \phi(n)$ bits of advice, in time $f(n)$, where $f$ is a polynomial. Without loss of generality assume that $f(n)= n^a-2$  for some constant $a \in \mathbb{N}$.
Since $\phi$ diverges to infinity, there exists some $n_0 \in \mathbb{N}$ such that $2^{\phi(n)-2} >\log a$, for all $n \ge n_0$.
There are at most $ \frac{\log \log n}{2^{\phi(n)-2}}$ binary strings of length at most $ \lfloor\log \log \log n - \phi(n) \rfloor$. Define $z= {\frac{ 2^{\phi(n)-2}}{{\log \log n}}}$.

Take a family of $(\lceil \frac{\log \log n}{2^{\phi(n)-2}}\rceil+1)$ oriented rings $C_{t_i}$ with $t_i$ nodes, for $ 0 \le i \le \lceil\frac{\log \log n}{2^{\phi(n)-2}}\rceil$, where  $t_i=\lfloor n^{({\log n}) ^{iz}}\rfloor$ and $n \ge n_0$.
By the pigeonhole principle, there exist indices $i,j$ such that  $C_{t_i}$ and $C_{t_j}$ must have the same advice string, with $t_i < t_j$.
The correctness of $\cA$ implies that the agent explores all the nodes of $C_{t_i}$ and stops after $f(t_i)$ steps. Since the agent has the same advice for $C_{t_j}$, it also stops after $f(t_i)$ steps in $C_{t_j}$.

We have $t_j  \ge t_{i+1} \ge  t_i ^{({\log n}) ^{z}}-1$.
Since $2^{\phi(n)-2} =\frac{2^{\phi(n)-2}}{{(\log \log n)}} \log \log n > \log a$, we have
$\log (({\log n}) ^{\frac{2^{\phi(n)-2}}{{(\log \log n)}}})$ $ > \log a$. Hence, ${({\log n}) ^{z}} >a$. Therefore, $t_j > {t_i}^a -1=  f(t_i)+1$, i.e., $t_j -1> f(t_i)$. This implies that the agent stops after fewer than $t_j-1$ steps in $C_{t_j}$, but it is not possible to explore $C_{t_j}$ in fewer than $t_j-1$ steps, which is a contradiction.
Therefore, there does not exist any algorithm that can explore an $n$-node oriented ring with advice of size at most $ \log  \log \log n - \phi(n)$.
\end{proof}

Notice that  Lemmas \ref{poly-ub} and \ref{poly-lb} hold both for the instance oracle and for the map oracle. The positive result from  Lemma \ref{poly-ub}
holds even for the map oracle, as the advice concerns the size of the graph and does not require knowing the starting node of the agent. The negative result from  Lemma \ref{poly-lb}
holds even for the instance oracle, as it is true even in oriented rings, where knowledge of the starting node does not provide any insight, since all nodes look the same.
Hence  Lemmas \ref{poly-ub} and \ref{poly-lb} imply the following theorem that gives a precise answer to the question stated at the beginning of this section.

\begin{theorem}
The minimum size of advice permitting the agent to explore any graph in time polynomial in the size $n$ of the graph is $\log\log\log n -\Theta(1)$, both
 for the instance oracle and for the map oracle.
\end{theorem}

\section{Fast exploration}

When advice given to the agent can be large,  there are two natural time thresholds:  $\Theta(n^2)$ for a map oracle, and $\Theta(n)$ for an instance oracle.
This is because, in both cases, these time benchmarks can be achieved  with sufficiently large advice. Indeed, we have the following proposition.

\begin{proposition}\label{ub}
1. There exists an exploration algorithm, working in time $O(n^2)$ and using advice of size $O(n\log n)$, provided by a map oracle, for $n$-node graphs.\\
2. There exists an exploration algorithm, working in time $O(n)$ and using advice of size $O(n\log n)$, provided by an instance oracle, for $n$-node graphs.
\end{proposition}

\begin{proof}
In order to guarantee exploration time $O(n^2)$, it is enough to provide the agent with some port-numbered spanning tree $T$ of the graph.
Given such a tree $T$,  the agent
identifies an Euler tour $E(u)$ of $T$ starting at $u$, for any node $u$ of $T$. (The tour is coded as a sequence of port numbers, including both the outgoing port and the incoming port at every step.) Let $\overline{E}(u)$ be the reverse string of $E(u)$, and let  $F(u)$ be the concatenation of $E(u)$ and  $\overline{E}(u)$.
The agent performs each $F(u)$ one after another. If some tour is impossible to continue, because the required outgoing port is not available, or the incoming port is not matched (which can happen, if $u$ is not the starting node of the agent), the tour is simply aborted, and the agent backtracks to its starting node. The tour $F(v)$,
where $v$ is the starting node of the agent, must succeed, and this tour visits all nodes of $T$, and hence explores the graph. Since there are $n$ tours, each of length $O(n)$, this gives time $O(n^2)$.

In order to guarantee exploration time $O(n)$, it is enough to provide the agent with some rooted port-numbered spanning tree $T$ of the graph, where the root is the starting node of the agent. The agent performs an Euler tour of this rooted tree in time $2n-2$.

It remains to show that a port-numbered spanning tree can be given by a map oracle,  and a port-numbered spanning tree rooted at the starting node of the agent
can be given by an instance oracle, in both cases using  $O(n\log n)$ bits.
This can be done as follows. Consider the DFS tree $T$ rooted at some arbitrary node in the case of a map oracle, and rooted at the starting node of the agent, in the case of an instance oracle, where neighbors of a node are explored in order of increasing port numbers. Give the shape of $T$ as a binary sequence of length $2n-2$,
where a 0 means ``go down the tree'' and 1 means ``go up the tree'', in this DFS exploration. Moreover, give the sequence of all port numbers, as they are encountered
in order during this DFS exploration. This is a sequence of $O(n)$ terms each of which is a number smaller than $n$. Hence, both sequences can be given to the agent
using $O(n\log n)$ bits. On their basis, the agent first reconstructs the shape of $T$, and then puts the port numbers in appropriate places.
\end{proof}

In the rest of this section we prove negative results indicating the quality of the natural solution given in Proposition  \ref{ub}. For the map oracle, we show that
quadratic exploration time cannot be beaten, and we give a lower bound on the size of advice sufficient to guarantee this time. For the instance oracle, we show that Proposition  \ref{ub} gives optimal advice for linear exploration time.

\subsection{Map oracle}\label{sec:map}

Our first result for the map oracle shows that, regardless of the size of advice, exploration time $\Theta(n^2)$ cannot be beaten, for some $n$-node graphs.

We will use the following construction from \cite{BRT} of a family  $\cH_X$ of graphs.

Let $H$ be an $\frac{m}{2}$-regular graph with $m$ nodes, where $m$ is even, e.g., the complete bipartite graph. Let $T$ be the set of edges of any spanning tree of $H$. Let $S$ be the set of edges of $H$ outside $T$. Let $s=|S|=\frac{m^2}{4}-m+1$ and  $S=\{e_1,e_2, \cdots, e_s\}$.

For $x \in \{0,1\}^s \setminus \{ 0^s\}$, the $(2m)$-node graph $H_x$ is constructed from $H$ by taking two disjoint copies $H'$ and $H''$ of $H$, and crossing some pairs of edges  from one copy to the other.  For $i=1,\cdots ,s$, if the $i-$th bit of $x$ is 1, then the edge $e_i=(u_i,v_i)$ is deleted from both copies of $H$ and two copies of $e_i$ are crossed between the two copies of $H$.
More precisely, let $\{v_1,\cdots ,v_m\}$ be the set of nodes of $H$ and let $v'_i$ and $v''_i$ be the nodes corresponding to $v_i$, in $H'$ and $H''$, respectively.
Let $V'$ and $V''$ be the sets of nodes of $H'$ and $H''$, respectively.
Define $H_x= (V'\cup V'', E_x)$, where  $E_x=\{(v'_i,v'_j), (v''_i,v''_j) : (v_i,v_j) \in T\} \cup \{(v'_i,v''_j), (v''_i,v'_j) : e_k=(v_i,v_j) \in S{ ~~and ~~} x_k=1 \}$.
Let $\cH_X=\{H_x : x \in \{0,1\}^s \setminus \{ 0^s\}\}$.

According to the result from \cite{BRT}, for every node $v \in H$, there exists some sequence $x(v) \in \{0,1\}^s \setminus \{ 0^s\}$ such that if
an exploration of $H$ performed according to some sequence $W$ of port numbers, starting from node $v_1$, visits node $v$ at most $s$ times,
then in one of the copies $H'$ or $H''$ in $H_{x(v)}$ the node $v'$ or $v''$ is not visited at all, if the same sequence $W$ is used to explore the graph $H_{x(v)}$ starting
from $v'_1$. Intuitively, the result from \cite{BRT} shows a class of graphs with the property that if some node in one of these graphs is not visited many times, then the exploration algorithm fails {\em in some other graph} of this class. There is no control in which graph of the class this will happen.  We use the graphs  from \cite{BRT}
as building blocks to prove a different kind of lower bound. Indeed, we construct a {\em single graph} having the property that if some of its nodes is not visited many times, then exploration must fail in this graph. This will prove a lower bound on exploration time for some graph, even if the agent knows the entire graph.

Using the graphs $H_x \in \cH_X$ from \cite{BRT} we construct the graph
$\widehat{G}$ as follows.
For any  $1 \le i \le m$, let $ v'_1(i)$ be the node corresponding to node $ v'_1$ from $H'$ in the graph $H_{x(v_i)}$.
Connect the graphs $H_{x(v_i)}$, for $1 \le i \le m$, and an oriented cycle $C$ with nodes $\{y_1, \cdots, y_{m}\}$
(port numbers 0,1 are in clockwise order at each node of the cycle), by adding edges $(y_i, v'_1(i))$,
for $1 \le i \le m$. The  port numbers corresponding to these edges are:
2 at $y_i$   and $\frac{m}{2}$ at $v'_1(i)$. See Fig. \ref{graph G^}. The cycle $C$ is called the {\em main cycle} of $\widehat{G}$.

\begin{figure}[h]
\centering
\includegraphics[width=0.4\textwidth]{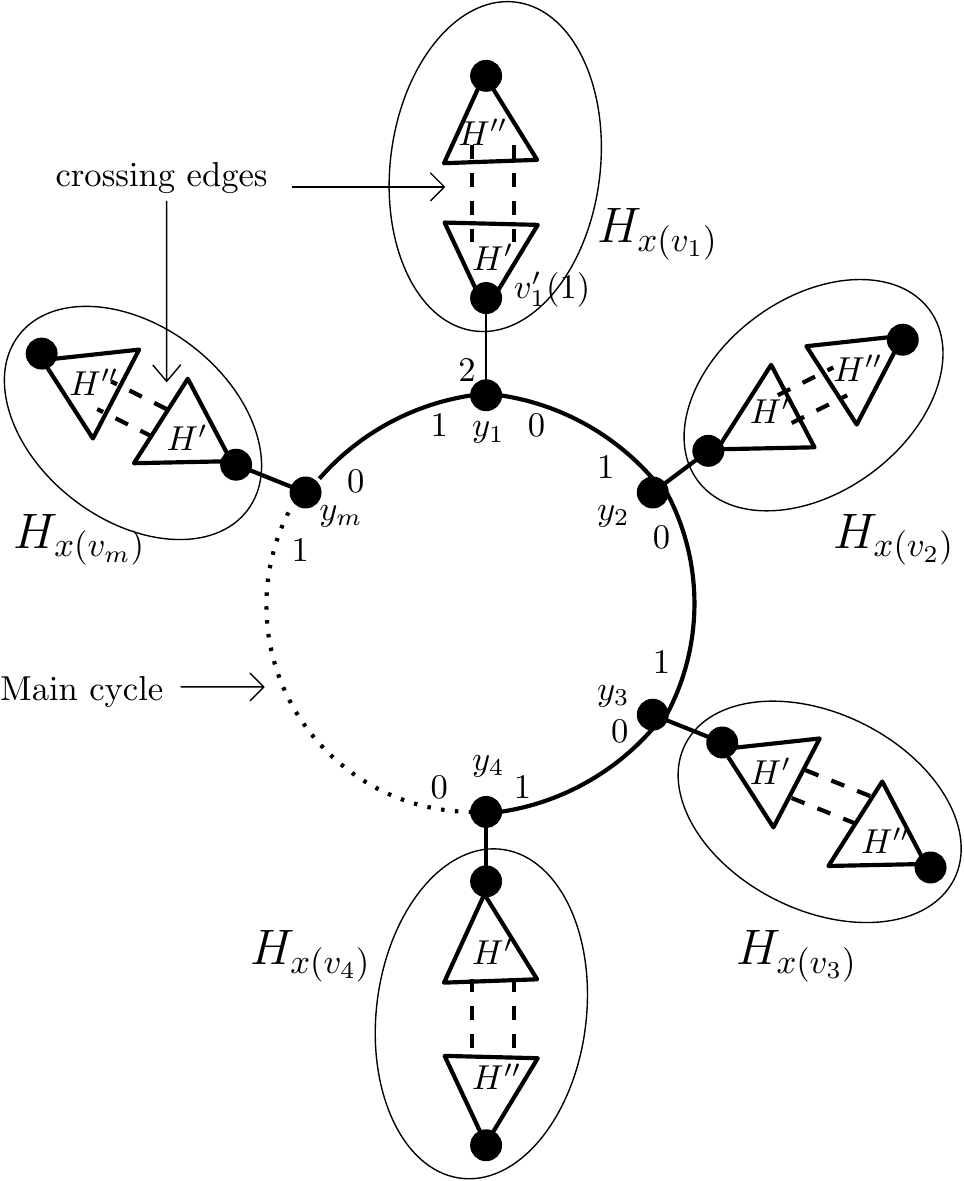}
\caption{Construction of $\widehat{G}$}
\label{graph G^}
\end{figure}

Let $n=2m^2 +m$ be the number of nodes in $\widehat{G}$.

By the construction of $\widehat{G}$, any exploration algorithm with the agent starting from any node of the main cycle, has the following {\em obliviousness property}.
For any step $i$ of the algorithm, if the agent is at some node $v$ in this step, and the algorithm prescribes taking some port $p$ at this node, then the port $q$
through which the agent enters the adjacent node $w$ in the $(i+1)$-th step, and the degree of the node $w$ are predetermined (i.e., they are independent of the starting node in the main cycle).  Intuitively, the agent does not learn anything during the algorithm execution.
Therefore, every exploration algorithm with the agent starting from any node of the main cycle can be uniquely coded by a sequence of port numbers which the agent takes in consecutive steps of its exploration.

Let $\cB$ be any exploration algorithm for $\widehat{G}$,  and suppose that the agent starts from some node of the main cycle.
We use $\cdot$ for concatenation of sequences.
\begin{lemma}\label{repr}
Let $U$ be the sequence of port numbers corresponding to the movement of the agent according to algorithm $\cB$, starting at some node of the main cycle $C$ of $\widehat{G}$ . Then $U= B'_1 \cdot (2) \cdot B_1 \cdot  (\frac{m}{2}) \cdot B'_2 \cdot (2)  \cdot B_2 \cdot (\frac{m}{2}) \cdots B'_p \cdot (2) \cdot B_p  \cdot (\frac{m}{2}) \cdot B'_{p+1}$, where each $B'_j$ is a sequence of port numbers corresponding to the movement of the agent along $C$ and each $B_j$ is a sequence of port numbers corresponding to the movement of the agent inside some $H_{x(v_i)}$.
\end{lemma}

\begin{proof}
Let the agent start from $y_{i_0}$. Let $y_{i_1}$ be the first node where the agent takes the port 2. Let $B'_1$ be the sequence of port numbers corresponding to the movement from $y_{i_0}$ to $y_{i_1}$. At $y_{i_1}$, the agent takes the port 2. Let $B_1$ be the sequence of port numbers corresponding to the movement of the agent after it takes the port 2  and before it takes the port $\frac{m}{2}$ at $v_1'(i_1)$. Therefore the sequence of port numbers until this moment can be written as $B'_1\cdot (2) \cdot  B_1 \cdot  (\frac{m}{2})$.
Continuing in this way, the sequence $U$ of port numbers can be written as $B'_1 \cdot (2) \cdot B_1 \cdot  (\frac{m}{2}) \cdot B'_2 \cdot (2)  \cdot B_2 \cdot (\frac{m}{2}) \cdots B'_p \cdot (2) \cdot B_p  \cdot (\frac{m}{2}) \cdot B'_{p+1}$.
 \end{proof}

%

Call an exploration algorithm of $\widehat{G}$ {\em non-repetitive}, if the agent, starting from the main cycle, enters each $H_{x(v_i)}$, for  $1 \le i \le m$, exactly once.
By definition, the sequence of port numbers corresponding to a non-repetitive algorithm can be written as
$D'_1 \cdot (2) \cdot D_1 \cdot  (\frac{m}{2}) \cdot D'_2 \cdot (2)  \cdot D_2 \cdot (\frac{m}{2}) \cdots D'_m \cdot (2) \cdot D_m  \cdot (\frac{m}{2}) \cdot D'_{m+1}$, where  each $D'_j$ is a sequence of port numbers corresponding to the movement of the agent along $C$ and each $D_j$ is a sequence of port numbers corresponding to the movement of the agent inside some $H_{x(v_i)}$. Notice that since the algorithm is non-repetitive, the number of blocks $D_j$ is exactly $m$.

 The following lemma proves that to show a lower bound on the exploration time in $\widehat{G}$, it is enough to consider only the class of non-repetitive algorithms.
\begin{lemma}\label{non-rep}
If the agent starts from some node of the main cycle of $\widehat{G}$ and executes
any exploration algorithm $\cB$ of $\widehat{G}$, then there exists a non-repetitive algorithm $\cB'$ for this agent, such that the exploration time of $\cB'$ is at most
the exploration time of $\cB$.
\end{lemma}
\begin{proof}
$U= B'_1 \cdot (2) \cdot B_1 \cdot  (\frac{m}{2}) \cdot B'_2 \cdot (2)  \cdot B_2 \cdot (\frac{m}{2}) \cdots B'_p \cdot (2) \cdot B_p  \cdot (\frac{m}{2}) \cdot B'_{p+1}$ be the sequence of port numbers corresponding to the algorithm $\cB$, represented as in Lemma \ref{repr}. For any
 $1 \le j \le m$, consider all the blocks $B_i$ corresponding to moves of the agent inside $H_{x(v_j)}$, and let  $D_j$ be the concatenation of all these blocks. Delete all these blocks together with the preceding $(2)$ and succeeding $(\frac{m}{2})$ for each block. Put $(2) \cdot D_j \cdot (\frac{m}{2}) $ at the position of the first deleted block.
 Consider the sequence $U'$, which results from this operation performed for all  $1 \le j \le m$. By definition, the sequence $U'$ of port numbers
 corresponds to a non-repetitive algorithm, and the length of $U'$ is less or equal to the length of $U$. Hence the lemma follows.
\end{proof}

The next lemma implies that the sequence $U$ corresponding to a correct non-repetitive exploration algorithm must be long.

\begin{lemma}
Let $U= D'_1 \cdot (2) \cdot D_1 \cdot  (\frac{m}{2}) \cdot D'_2 \cdot (2)  \cdot D_2 \cdot (\frac{m}{2}) \cdots D'_m \cdot (2) \cdot D_m  \cdot (\frac{m}{2}) \cdot D'_{m+1}$ be the sequence of port numbers corresponding to a non-repetitive algorithm. If there exists some $D_i$ such that the agent following $D_i$ in $H$ starting from node $v_1$  visits some node $v_j$ of $H$ at most $s$ times, then there exists a starting node in the main cycle of $\widehat{G}$, such that the agent starting at this node and following $U$ does not visit all the nodes of $\widehat{G}$.
\label{lem:lemmap1}\end{lemma}
\begin{proof}
 Suppose that there exists some $D_i$ such that the agent following $D_i$ in $H$ starting from node $v_1$  visits some node $v_j$ of $H$ at most $s$ times. Choose the starting node of the agent in the main cycle of $\widehat{G}$,  so that the part of its trajectory corresponding to  $D_i$ visits $H_{x(v_j)}$. Then by the property of $H_{x(v_j)}$, at least one copy of $v_j$ in
$H_{x(v_j)}$ will not be explored by the agent. Since $H_{x(v_j)}$ is  visited by the agent only when it follows $D_i$, some node in $\widehat{G}$  is not explored by the agent following~$U$.
\end{proof}

\begin{theorem}
Any exploration algorithm using any advice given by a map oracle must take time $\Omega(n^2)$ on graph $\widehat{G}$, for some starting node in the main cycle,
for arbitrarily large $n$.
\label{th:thmap}\end{theorem}

\begin{proof}
By Lemma \ref{non-rep}, it is enough to consider only non-repetitive algorithms.
Let $U= D'_1 \cdot (2) \cdot D_1 \cdot  (\frac{m}{2}) \cdot D'_2 \cdot (2)  \cdot D_2 \cdot (\frac{m}{2}) \cdots D'_m \cdot (2) \cdot D_m  \cdot (\frac{m}{2}) \cdot D'_{m+1}$ be the sequence of port numbers  corresponding to such an algorithm. Then by Lemma \ref{lem:lemmap1},
for each $i$, $1 \le i \le m$, the agent following $D_i$ in $H$ starting from node $v_1$ visits each node $v_j$ of $H$, for $1 \le j \le m$, at least $s+1$ times.
Therefore the length of $D_i$ is at least $(s+1)m$. Hence,  the length of $U$ is at least $\sum_{i=1}^m (s+1)m= (s+1)m^2=m^2(\frac{m^2}{4} -m +1)$. Since $n=2m^2+m$, the length of $U$ is in $\Omega(n^2)$.
\end{proof}

%

Theorem \ref{th:thmap} shows that, for some $n$-node graph,  no advice given by a map oracle can help to explore this graph in time better than $\Theta(n^2)$. It is then natural to ask what is the minimum size of advice
 to achieve time $\Theta(n^2)$ with a map oracle, for every $n$-node graph. Our next result shows that any exploration algorithm using advice of size $n^{\delta}$ for $\delta < \frac{1}{3}$, must take time $\omega(n^2)$, on some $n$-node graph.

Fix a constant $\epsilon <\frac{1}{2}$.
Let $H$ be an $\frac{m}{2}$-regular graph with $m$ nodes, where $m$ is even. Let $\{v_1,\cdots ,v_m\}$ be the set of nodes of $H$. Consider a subset $Z \subset \{ 1,2, \cdots, m\}$ of size $m^{\epsilon}$. Let $p=m^{\epsilon}$ and $n=2mp+p$.
We construct an $n$-node graph $\widehat{G_Z}$ from $H$. The construction of $\widehat{G_Z}$ is similar to the construction of $\widehat{G}$ at the beginning of this section.
Let $Z=\{z_1,z_2, \cdots, z_p\}$. To construct $\widehat{G_Z}$, connect the (previously described) graphs $H_{x(v_{z_i})}$, for $1 \le i \le p$, and an oriented cycle $C'$ 
(called the main cycle) with nodes $\{y_1, \cdots, y_{p}\}$, by adding edges $(y_i, v'_1(z_i))$, for $1 \le i \le p$. The  port numbers corresponding to these edges are:
2 at $y_i$   and $\frac{m}{2}$ at $v'_1(z_i)$. Note that the same obliviousness property applies to exploration algorithms in graphs $\widehat{G_Z}$, when the agent
starts from a node of the main cycle.

Let $\widehat{\cG_Z}$ be the set of all possible graphs $\widehat{G_Z}$ constructed from $H$. We have $|\widehat{\cG_Z}|={m \choose p}$.


\begin{theorem}\label{exceed}
For any $\epsilon < \frac{1}{2}$, any exploration algorithm  using advice of size $o(n^{\frac{\epsilon}{1+\epsilon}} \log n)$ must take time $\omega(n^2)$ on some  graph
of the class $\widehat{\cG_Z}$ and for some starting node in the main cycle of this graph, for arbitrarily large $n$.
\end{theorem}
\begin{proof}
Since $\epsilon < \frac{1}{2}$, there exists an integer  $c$ such that $\epsilon < \frac{c-1}{2c-1}$.
We show that if the size of the advice is at most $\frac{1}{2c}m^{\epsilon} \log (m^{1-\epsilon})\leq \frac{1-\epsilon}{2c(1+\epsilon)}(\frac{n}{2})^{\frac{\epsilon}{1+\epsilon}} \log \frac{n}{2}$, then there exists a graph in $\widehat{\cG_Z}$ for which the time required for exploration is $\omega(n^2)$. We have $|\widehat{\cG_Z}|={m \choose {m^{\epsilon}}} \ge ({m^{1-\epsilon}})^{m^{\epsilon}}$. There are fewer than $({m^{1-\epsilon}})^{\frac{m^{\epsilon}}{c}}$ different binary strings of length at most $\frac{1}{2c}m^{\epsilon} \log (m^{1-\epsilon})$. By the pigeonhole principle, there exists a family of graphs $\widehat{\cG} \subset \widehat{\cG_Z}$ of size at least $({m^{1-\epsilon}})^{(c-1)\frac{m^{\epsilon}}{c}}$ such that all the graphs in $\widehat{\cG}$ get the same advice.

Define $F(\widehat{\cG})=\bigcup \left\{\{v_{z_1}, v_{z_2}, \cdots ,v_{z_p}\} : Z=\{z_1,z_2, \cdots ,z_p\} ~ {\rm and}~ \widehat{G_Z} \in \widehat{\cG} \right\}$. Intuitively, $F(\widehat{\cG})$ is the subset of nodes of $H$, such that for each $v \in F(\widehat{\cG})$, there exists some graph in $\widehat{\cG}$ that contains $H_{x(v)}$ as a subgraph.

{\bf Claim:} $|F(\widehat{\cG})| \ge {|\widehat{\cG}|}^{\frac{1}{p}}$.

We prove the claim by contradiction. Suppose that $|F(\widehat{\cG})| < {|\widehat{\cG}|}^{\frac{1}{p}}$.
Each graph in $\widehat{\cG}$ has $p$ different subgraphs $H_{x(v)}$, where $v \in |F(\widehat{\cG})|$. There are ${|F(\widehat{\cG})| \choose p}$ different graphs in $\widehat{\cG}$ which is at most $|F(\widehat{\cG})|^p < |\widehat{\cG}|$. This  contradiction proves the claim.

Consider the exploration of some graph $\widehat{G_Z} \in \widehat{\cG}$ starting from the main cycle. Let $U= D'_1 \cdot (2) \cdot D_1 \cdot  (\frac{m}{2}) \cdot D'_2 \cdot (2)  \cdot D_2 \cdot (\frac{m}{2}) \cdots D'_p \cdot (2) \cdot D_m  \cdot (\frac{m}{2}) \cdot D'_{p+1}$ be the sequence of port numbers  corresponding to  a non-repetitive algorithm exploring $\widehat{G_Z}$. Then for each $i$, $1 \le i\le p$, the agent following $D_i$ in $H$ starting from node $v_1$ must visit each node $v \in F(\widehat{\cG})$  at least $s+1$ times.
(Otherwise, there would exist a graph in $\widehat{\cG}$ and a starting node in the main cycle, for which one node would not be explored by $U$). Hence, for sufficiently large $m$, the length of $D_i$ is at least $(s+1)|F(\widehat{\cG})| \ge \frac{m^2}{5} m^{\frac{(c-1)(1-\epsilon)}{c}}$, because $s \geq  \frac{m^2}{5}$. Therefore, the length of $U$ is at least $p \frac{m^2}{5} m^{\frac{(c-1)(1-\epsilon)}{c}}=\frac{1}{5} m^{\epsilon} m^2 m^{\frac{(c-1)(1-\epsilon)}{c}}=\frac{1}{5}m^{2+\epsilon+\frac{(c-1)(1-\epsilon)}{c}}=\frac{1}{5}m^{2+2\epsilon+(\frac{c-1}{c}+\frac{1-2c}{c}\epsilon)}$. Since $\epsilon < \frac{c-1}{2c-1}$, we have $(\frac{c-1}{c}+\frac{1-2c}{c}\epsilon)>0$. Therefore, the length of $U$ is in $\omega(m^{2+2\epsilon})=\omega(n^2)$, and hence exploration time is in $\omega(n^2)$.
\end{proof}

Since $\epsilon <\frac{1}{2}$ implies $\frac{\epsilon}{1+\epsilon} <\frac{1}{3}$, Theorem \ref{exceed} yields the following corollary.

\begin{corollary}
For any $\delta <\frac{1}{3}$, any exploration algorithm  using advice of size $o(n^{\delta})$ must take time $\omega(n^2)$ on some $n$-node graph,
for arbitrarily large $n$.
\end{corollary}

\subsection{Instance oracle}\label{sec:instance}

For the instance oracle we show a general lower bound on the size of advice needed to achieve a given exploration time. The main corollaries of this lower bound are:
\begin{itemize}
\item
the size of advice $\Theta(n\log n)$ from Proposition \ref{ub}, sufficient to achieve linear exploration time, cannot be beaten;
\item
for advice of linear size, exploration time must be quadratic.
\end{itemize}

To prove our lower bound we will use the following construction.

Let $G$ be an $\frac{n}{4}$-regular $\frac{n}{2}$-node graph, where $n$ is divisible by 4. We can use, for example, the complete bipartite graph with $\frac{n}{2}$ nodes.
Let $m=\frac{n}{2}$. Let $v_1$, $v_2$, $\cdots$, $v_m$ be the nodes of $G$.
Let $x=(x_1,x_2, \cdots, x_m)$ be a sequence of $m$ integers where $0 \le x_i \le \frac{m}{2}-1$, for $i=1,\cdots,m$. Let $X$ be the set of all such sequences.

\begin{figure}
\begin{subfigure}{.5\textwidth}
  \centering
  \includegraphics[width=.6\linewidth]{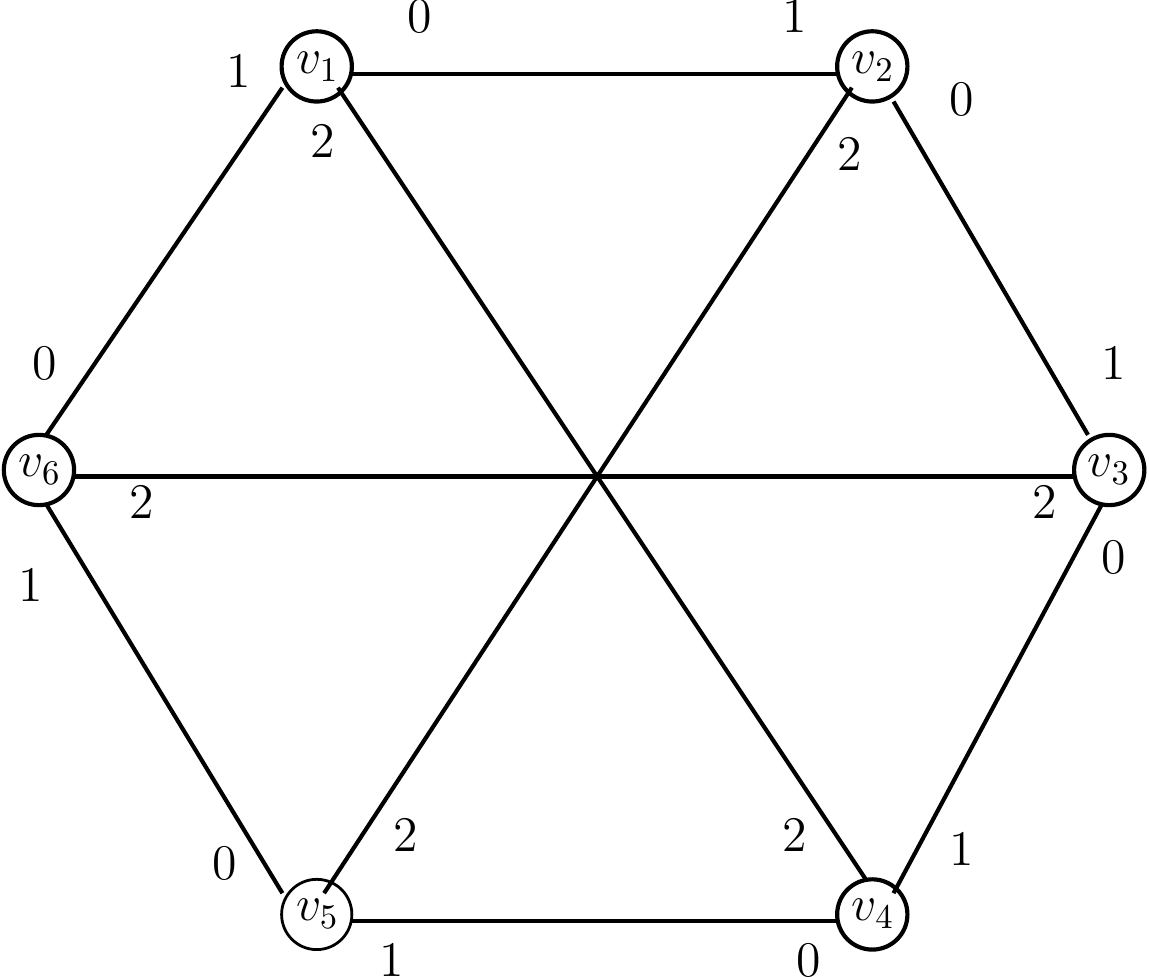}
  \caption{An example of $G$ with six nodes}
\end{subfigure}%
\hspace{-1cm}
\begin{subfigure}{.5\textwidth}
  \centering
  \includegraphics[width=1\linewidth]{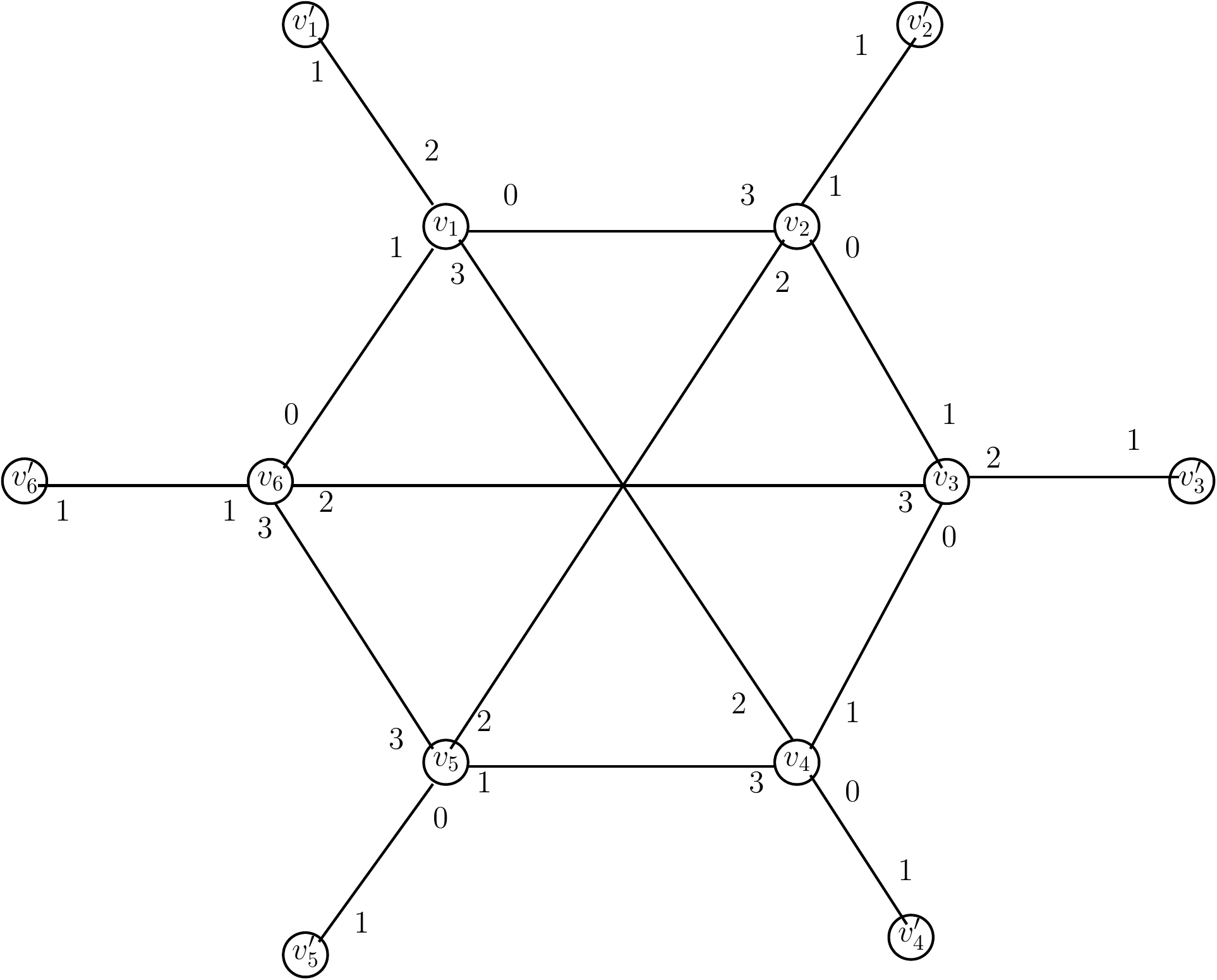}
  \caption{$G_x$ for the sequence $x=(2,1,2,0,0,1)$}
\end{subfigure}
\caption{The construction of $G_x$ from $G$}
\label{fig:fig1}
\end{figure}

 We construct an $n$-node graph $G_x$ as follows. For each $i=1,\cdots,m$, add a new node $v_i'$ of degree 1 to $G$. Replace the port number $x_i$ at $v_i$ by port number $\frac{m}{2}$. Add the edge $(v_i,v_i')$ with the port number $x_i$ at $v_i$. An example of the construction of $G_x$ from $G$ is shown in Fig. \ref{fig:fig1}. Let $\cG_X$ be the set of all possible graphs $G_x$ constructed from $G$.

\begin{theorem}\label{th:instance1}
For any function $\phi: \mathbb{N} \longrightarrow \mathbb{N}$, and for any exploration algorithm using advice of size $o(n\phi(n))$,
this algorithm must take time $\Omega(\frac{n^2}{2^{\phi(n)}})$ on some $n$-node graph from the family $\cG_X$, for arbitrarily large integers $n$.
\end{theorem}

\begin{proof}
Let $n$ be divisible by 4.
We show that if the size of the advice is at most $\frac{n\phi(n)}{4}-1$, then there exists an $n$-node graph in the family  $\cG_X$,  for which the time required for exploration is $\Omega(\frac{n^2}{2^{\phi(n)}})$.
We have $|\cG_X|=|X|= (\frac{m}{2})^m$. There are fewer than $2^{\frac{m\phi(2m)}{2}}=( 2^{\phi(2m)})^{\frac{m}{2}}$ different binary strings of length at most $\frac{2m\phi(2m)}{4}-1=\frac{n\phi(n)}{4}-1$. By the pigeonhole principle, there exists a family of graphs $\cG \subset \cG_X$, of size at least $\frac{(\frac{m}{2})^m}{{(2^{\phi(2m)})}^{\frac{m}{2}}}$, such that all the graphs in $\cG$ get the same advice. Let $Y=\{x \in X : G_x \in \cG\}$.
Let $z=\frac{m}{2^{\phi(2m)+2}}$ and let  $J= \{j : |  \{x_j: x \in Y\}| \ge z\}$. Intuitively,  $J$ is the set of indices, for which the set of terms of sequences $x$ that produce graphs from $\cG$ is large. Let $p=|J|$.

\noindent{\bf Claim:} $p > \frac{m}{2}$.

We prove the claim by contradiction. Suppose that $p \le \frac{m}{2}$.
Since $p\leq \frac{m}{2}$ and $z<\frac{m}{2}$, we have $(\frac{m}{2})^p \cdot z^{m-p}\leq (\frac{m}{2})^{\frac{m}{2}} \cdot z^{\frac{m}{2}} $.
Note that for all $i \in \{1, 2,\cdots, m\} \setminus J$, we have $| \{x_i: x \in Y\}| <z$, and for $j \in J$, we have $ | \{x_j: x \in Y\}|  \le \frac{m}{2}$.
Therefore, $|\cG| < (\frac{m}{2})^p \cdot z^{m-p}$.
Hence, $|\cG| < (\frac{m}{2})^{\frac{m}{2}} (\frac{m}{2^{\phi(2m)+2}})^{\frac{m}{2}} = \frac{m^m}{2^{\frac{m}{2}} {(2^{\phi(2m)+2})}^{\frac{m}{2}}} < |\cG|$, which is a contradiction. This proves the claim.

Consider any  exploration algorithm for the class $\cG$. There exists a graph $G_x \in \cG$, such that, at each node $v_j$ of $G_x$, for $j \in J$, the agent must take all the ports in $\{x_j: x \in Y\} $.
Indeed, suppose that the agent does not take some port  $x_j$, where $j\in J$ and $x \in Y$.
Consider the exploration of any graph $G_{x'} \in \cG$, where $x'_j=x_j$. Since the agent can visit $v_j'$ only coming from $v_j$, using port $x'_j$ in $G_{x'}$, the node $v_j'$ remains unexplored, as the port $x'_j$ at $v_j$ is never used, which is a contradiction.
Hence, the agent must visit at least $\frac{m}{2^{\phi(2m)+2}}$ ports at each node $v_j$ for $j \in J$. Since $|J| > \frac{m}{2}$, the time required for exploration is at least $\frac{m^2}{2^{\phi(2m)+3}}$, i.e., it is at least $\frac{n^2}{2^{\phi(n)+5}}$.
\end{proof}

If $\phi(n)=c$ where $c$ is a constant, then Theorem \ref{th:instance1} implies that any exploration algorithm using advice of size at most $\frac{cn}{2}$, must take time at least $\frac{n^2}{2^{c+3}}$. This implies that, if the size of advice is at most $c'n$, for any constant $c'$, then exploration time is $\Omega(n^2)$. Hence we have the following corollary.
\begin{corollary}
Any exploration algorithm using advice of size $O(n)$ must take time $\Omega(n^2)$ on some $n$-node graph, for arbitrarily large $n$.
\end{corollary}

For $\phi(n)\in o(\log n)$, Theorem \ref{th:instance1} implies an exploration time $\omega(n)$ which shows that the upper bound on the size of advice from Proposition \ref{ub}
is asymptotically tight for exploration in linear time. The following corollary improves this statement significantly, showing that exploration time is very sensitive to the size of advice at the threshold $\Theta(n\log n)$ of the latter.

\begin{corollary}\label{th:instance2}
Consider any  constant $\epsilon <2$.
Any exploration algorithm using advice of size $o(n \log n)$ must take time $\Omega(n^{\epsilon})$,  on some $n$-node graph,
for arbitrarily large $n$.
\end{corollary}
\begin{proof}
If the size of advice is $o(n\log n)$, then it is $n \phi(n)$, where $\phi(n)=\frac{\log n}{f(n)}$, with
$f(n) \rightarrow \infty$ as $n \rightarrow \infty$. Theorem  \ref{th:instance1} implies that exploration time must be $\Omega\left(\frac{n^2}{2^\frac{\log n}{f(n)}}\right)=
\Omega\left(\frac{n^2}{n^\frac{1}{f(n)}}\right)$. Since, for any constant $\delta >0$, we have $n^{\frac{1}{f(n)}} \in O(n^{\delta})$, the corollary holds.
\end{proof}

\section{Exploration of hamiltonian graphs}

In this section we turn attention to hamiltonian graphs. These graphs have  a special feature from the point of view of exploration: with sufficiently large advice
of appropriate type, the agent can explore a hamiltonian graph without any loss of time, visiting each node exactly once, i.e., in time $n-1$, for $n$-node graphs.
Indeed, an instance oracle can give as advice the  sequence of port numbers along a hamiltonian cycle, from the starting node of the agent, and then the agent takes the prescribed ports in $n-1$ consecutive steps. Since it is enough to give $n-1$ port numbers, and the binary representation of each port number uses $O(\log n)$ bits, advice of size $O(n\log n)$, given by an instance oracle, suffices.

We show that neither the quality nor the size of advice can be decreased to achieve the goal of optimal exploration of hamiltonian graphs. To prove the first statement,
we show a graph which is impossible to explore in time $n-1$ when advice of any size is given by a map oracle. Indeed, we construct an $n$-node hamiltonian graph
for which even knowing the entire map of the graph (but not knowing its starting node) an agent must use time $\Omega(n^2)$ to explore it.  To prove the second statement, we construct a class  of $n$-node hamiltonian graphs for which advice of size $o(n\log n)$, even given by an instance oracle, is not enough to permit exploration of graphs in this class in time $n-1$. Indeed, we show more: any exploration algorithm using such advice must exceed the optimal time $n-1$ by a summand
$n^\epsilon$, for any $\epsilon <1$, on some graph of this class.

In order to prove the first result, we construct a $(3n)$-node hamiltonian graph $\tilde{G}$ from the $n$-node graph $\widehat{G}$ described in Section \ref{sec:map}.
First, we  consider an $\frac{m}{2}$-regular $m$-node hamiltonian graph $H$ (for example, the complete bipartite graph). Let $v_1, v_2, \cdots , v_m$ be the nodes of $H$ along a hamiltonian cycle. The graph $\widehat{G}$ is constructed from $H$ as described in Section \ref{sec:map}, where the hamiltonian path $(v_1, v_2, \cdots , v_m)$ is taken as the spanning tree $T$. We construct the hamiltonian graph $\tilde{G}$ from the graph $\widehat{G}$ as follows.
Denote by $d(v)$ the degree of node $v$ in $\widehat{G}$.
For each node $v$ in $\widehat{G}$, consider a cycle of  three nodes $v(1)$, $v(2)$, and $v(3)$, in $\tilde{G}$, with port numbers $3d(v), 3d(v)+1$
in clockwise order at each of these three nodes.
For each edge $(u,v)$ in $\widehat{G}$, such that the port numbers corresponding to this edge are $p$ at $u$ and $q$ at $v$, add, in $\tilde{G}$, the edges $(u(i), v(j))$, for $1 \le i,j \le 3$, with the following port numbers. The port numbers corresponding to edge $(u(i), v(j))$ are:  $p+(j-1)d(u)$ at $u(i)$ and  $q+(i-1)d(v)$ at $v(j)$,
see Fig. \ref{fig:HAM1}.

\begin{figure}
\begin{subfigure}{.5\textwidth}
  \centering
  \includegraphics[width=.6\linewidth]{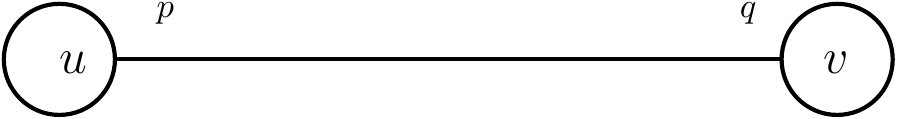}
\end{subfigure}%
\hspace{-1cm}
\begin{subfigure}{.5\textwidth}
  \centering
  \includegraphics[width=1\linewidth]{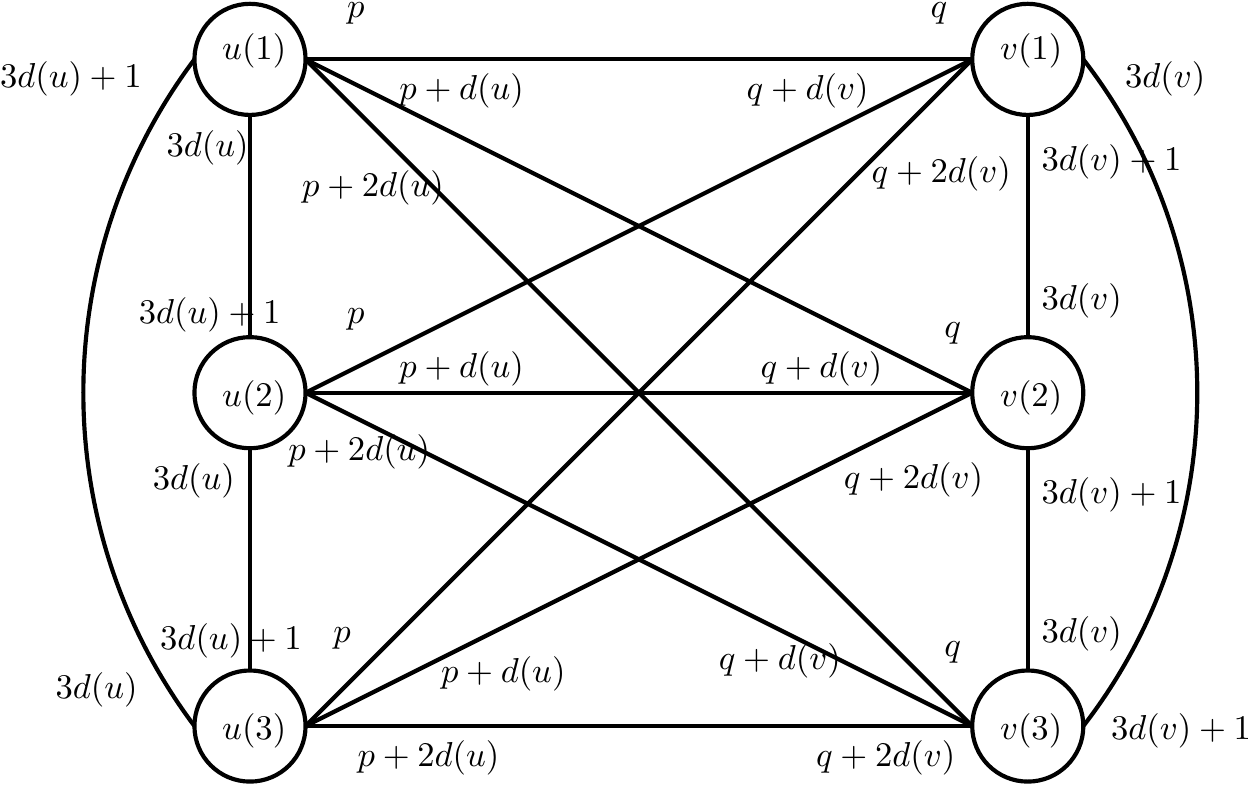}
\end{subfigure}
\caption{The construction of $\tilde{G}$ from $\widehat{G}$}
\label{fig:HAM1}
\end{figure}

\begin{lemma}
The graph $\tilde{G}$ is hamiltonian.
\end{lemma}
\begin{proof}
 According to the definition of $H_{x(v_i)}$ in Section \ref{sec:map}, this graph has a spanning tree $T_i$ that contains two copies of $T$, corresponding to $H'$ and $H''$, and an edge which crosses from $H'$ to $H''$. Consider the spanning tree $\widehat{T}$ of $\widehat{G}$ that contains all $T_i$, for $1 \le i \le m$, a spanning path $(y_1,\cdots ,y_m)$ of
 the main cycle, and the set of edges $(y_i, v'_1(i))$. See Fig. \ref{fig:HAM2}. Note that the maximum degree of a node in $\widehat{T}$ is three. Since each node in $\widehat{G}$ is replaced by a cycle of three nodes in $\tilde{G}$, an Euler tour of  the tree $\widehat{T}$ yields  a hamiltonian cycle in $\tilde{G}$. Therefore, $\tilde{G}$ is hamiltonian.
\end{proof}

\begin{figure}[h]
\centering
\includegraphics[width=0.7\textwidth]{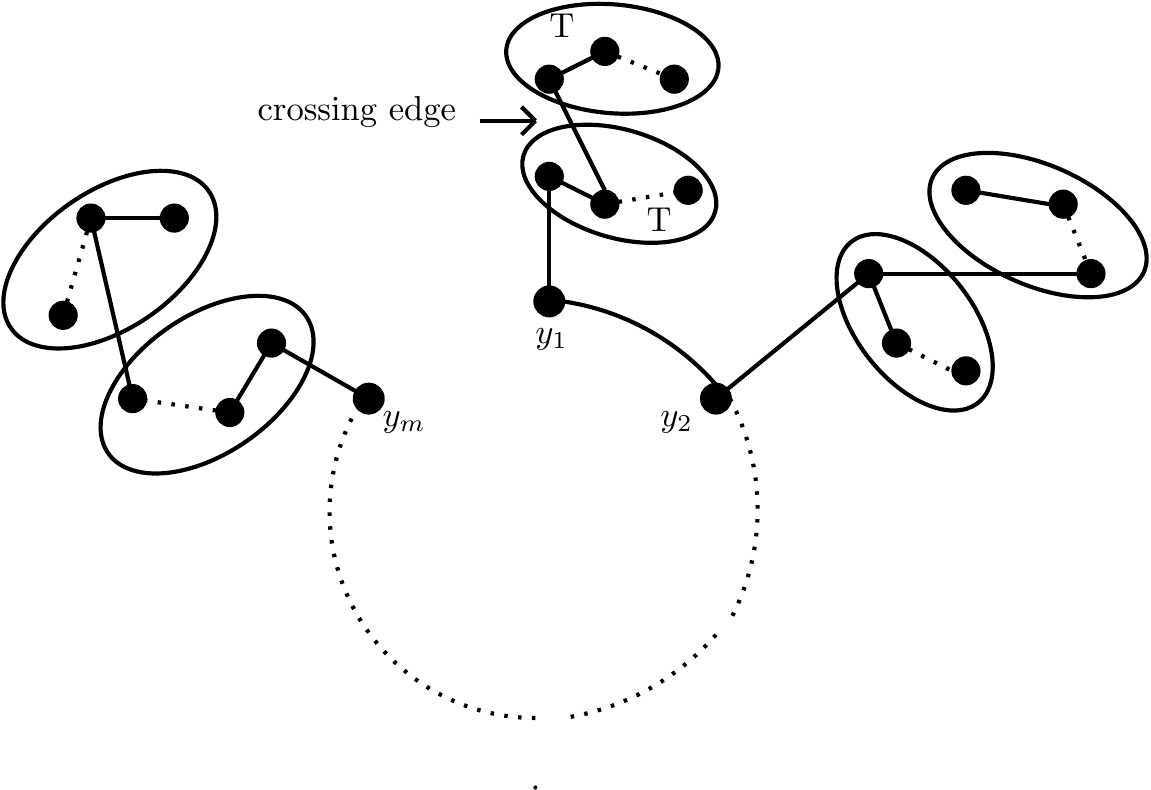}
\caption{The spanning tree $\widehat{T}$ of $\widehat{G}$}
\label{fig:HAM2}
\end{figure}
Let $\cB$ be an exploration algorithm for $\tilde{G}$ starting from node $y_i(1)$,  for some $i \leq m$. We describe the following algorithm $\cB^*$ on $\widehat{G}$,
starting from node $y_i$. Ignore all moves of $\cB$ taking port
$3d(v)$ or  $3d(v) +1$ at a node $v(j)$, for $1 \leq j \leq 3$,  of $\tilde{G}$. Replace every move of  $\cB$ taking port $r=p+(i-1)d(v)$, at node $v(j)$, for $1 \leq j \leq 3$, in $\tilde{G}$, where $0 \leq p \leq d(v)-1$, by  a move taking port $p$ in
$\widehat{G}$.

Then the agent executing $\cB ^*$ in $\widehat{G}$, starting from the main cycle, explores all the nodes. The time used by $\cB ^*$ in $\widehat{G}$ does not exceed the time used by
$\cB$ in $\tilde{G}$.
Since, by Theorem \ref{th:thmap}, any exploration algorithm for $\widehat{G}$, starting from the main cycle, must take time $\Omega(n^2)$,  algorithm $\cB$ must take time $\Omega(n^2)$  to explore $\tilde{G}$. Replacing $3n$ by $n$ we have the following theorem.

\begin{theorem}
Any exploration algorithm using any advice given by a map oracle must take time $\Omega(n^2)$ on some $n$-node hamiltonian graph, for arbitrarily large $n$.
\end{theorem}

Our last result shows that advice of size $o(n\log n)$ causes significant increase of exploration time for some hamiltonian graphs, as compared to optimal time $n-1$ achievable with advice of size $O(n\log n)$, given by an instance oracle.

\begin{theorem}
For any constant $\epsilon <1$, and for any
exploration algorithm using advice of size $o(n \log n)$, this algorithm must take time $n+ n^{\epsilon}$, on some $n$-node hamiltonian graph,
for arbitrarily large $n$.
\end{theorem}
\begin{proof}
Fix a constant $\epsilon <1$.
We show that if the size of the advice is at most $\frac{n(1-\epsilon)}{4} \log(\frac{n}{8})$, then there exists an $n$-node graph, for which the time required for exploration is $n+\Omega(n^{\epsilon})$.
The construction of the graphs is similar as in Section \ref{sec:instance}. Now, we start with a $\frac{n}{4}$-regular $\frac{n}{2}$-node hamiltonian graph $G$, for $n$ divisible by 4. We can use the complete bipartite graph with $\frac{n}{2}$ nodes, which is hamiltonian.
 We construct $G_x$ from $G$ as in Section \ref{sec:instance}. Let $v_1, v_2, \cdots , v_ \frac{n}{2}$ be the nodes of $G$ along a hamiltonian cycle. Let $G'_x$ be the graph which is obtained from $G_x$ by adding the set of edges $\{(v_{2i-1},v_{2i}) : 1 \le i \le \frac{n}{2}\}$. Each of these edges is given port number 2 at both its endpoints.
An example of the construction of $G'_x$ from $G$ is shown in Fig. \ref{fig:HAM3}.
Since $G$ is hamiltonian, $G'_x$ is also hamiltonian, as $(v_1,v_1',v_2',v_2, v_3, v_3', v_4', v_4, \cdots, v_\frac{n}{2},v_1)$ is a hamiltonian cycle in $G'_x$.

Let $\cG'_X$ be the set of all possible graphs $G'_x$ constructed from $G$. Then $|\cG'_X|=(\frac{m}{2})^m$, where $m=\frac{n}{2}$. There are fewer than $(\frac{m}{2})^{\frac{m(1-\epsilon)}{2}}$ different binary strings of length $\frac{m(1-\epsilon)}{2} \log(\frac{m}{4})=\frac{n(1-\epsilon)}{4} \log(\frac{n}{8})$. By the pigeonhole principle, there exists a family of graphs $\cG \subset \cG_X$, of size at least $(\frac{m}{2})^{m-\frac{m(1-\epsilon)}{2}}$, such that all the graphs in $\cG$ get the same advice.

Let $Y=\{x \in X : G'_x \in \cG'\}$.
Let $z=\frac{m^{\epsilon}}{4}$ and let  $J= \{j : |  \{x_j: x \in Y| \ge z\}$. Let $p=|J|$. By similar arguments as in the proof of Theorem \ref{th:instance1}, we have $p > \frac{m}{2}$. Therefore there exists some $i$, $1 \le i \le m$ for which both $v_{2i-1}$ and $v_{2i}$ belong to $J$. 

\begin{figure}
\begin{subfigure}{.5\textwidth}
  \centering
  \includegraphics[width=.6\linewidth]{fig1.pdf}
  \caption{An example of $G$ with six nodes}
\end{subfigure}%
\hspace{-1cm}
\begin{subfigure}{.5\textwidth}
  \centering
  \includegraphics[width=1\linewidth]{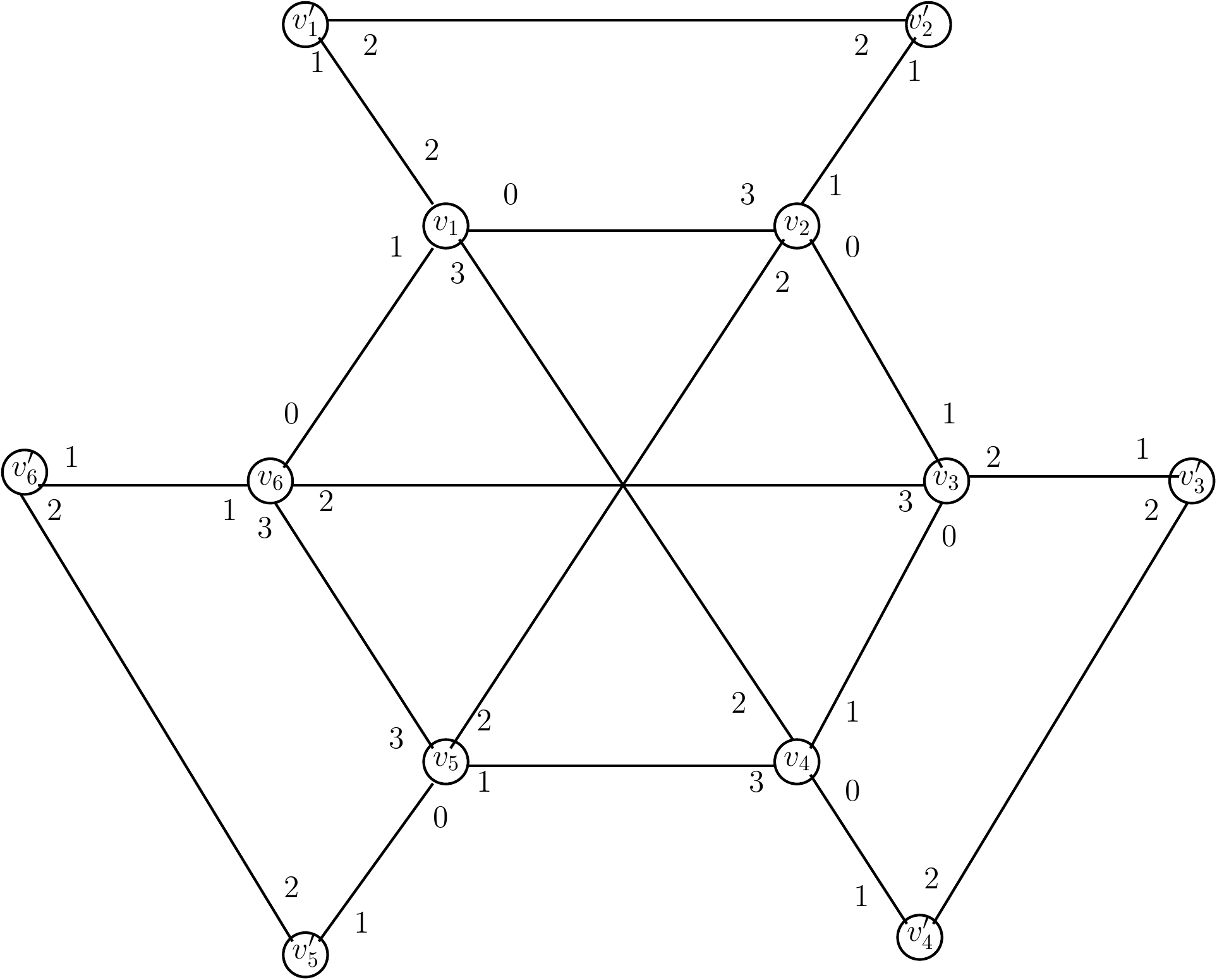}
  \caption{$G'_x$ for the sequence $x=(2,1,2,0,0,1)$}
\end{subfigure}
\caption{The construction of $G'_x$ from $G$}
\label{fig:HAM3}
\end{figure}

Consider any exploration algorithm for the class $\cG'$.
Let the agent start from some node $v_k$, $1 \le k \le m$. The agent must traverse one of the edges $(v_{2i-1}, v'_{2i-1})$ or $(v_{2i}, v'_{2i})$ in order to visit the nodes $v'_{2i-1}$ and $v'_{2i}$. Hence, by a similar argument as in the proof of Theorem \ref{th:instance1}, the agent must take at least $\frac{m^{\epsilon}}{4}$ ports at the node
$v_{2i-1}$ or $v_{2i}$, on some graph $G'_x \in \cG'$. To visit the remaining $(n-1)$ nodes, the agent needs time at least $n-2$. Hence, the time required for visiting all the nodes of $G'_{x}$ is at least $n-2 + \frac{m^{\epsilon}}{4}$, i.e., the time required for exploration is $n+ \Omega(n^{\epsilon})$.
Since this holds for any $\epsilon <1$, the simpler lower bound $n+ n^{\epsilon}$, for any $\epsilon <1$, holds as well, for arbitrarily large $n$.
\end{proof}

\section{Conclusion}

Most of our lower bounds on the size of advice are either exactly or asymptotically tight. The lower bound $\log\log\log n -\Theta(1)$ on the size of advice sufficient to explore all $n$-node graphs in polynomial time is exactly tight: with advice of any such size, polynomial exploration is possible, and with advice of any smaller size it is not.
For an instance oracle, the lower bound $\Omega (n\log n)$ on the size of advice sufficient to explore $n$-node graphs in $O(n)$ time is asymptotically tight, as  we gave a linear time exploration algorithm using advice of size $O(n\log n)$.
An exception to this tightness is the lower bound on the size of advice given by a map oracle, permitting exploration in time $O(n^2)$. While the natural upper bound
is $O(n\log n)$, our lower bound is only $n^{\delta}$ for any $\delta <\frac{1}{3}$. Hence the main remaining question is:

What is the smallest advice, given by a map oracle,  permitting exploration of $n$-node graphs in time $O(n^2)$?

{\bf Acknowledgements.}
We are grateful to Adrian Kosowski for early discussions on the subject of this paper and for drawing our attention to \cite{BRT}.


\begin{thebibliography}{12}

\bibitem{AKM01}
S.~Abiteboul, H.~Kaplan, T.~Milo, Compact labeling schemes for ancestor
queries, Proc. 12th Annual ACM-SIAM Symposium on Discrete
Algorithms (SODA 2001), 547--556.



\bibitem{AH}
S. Albers and M. R. Henzinger,
Exploring unknown environments,
SIAM Journal on Computing 29 (2000), 1164-1188.

\bibitem{AKLLR}
R. Aleliunas, R. Karp, R. Lipton, L. Lovasz C. Rackoff, Random walks, universal traversal sequences, and the complexity of maze problems,
Proc. 20th Annual IEEE Symposium on Foundations of Computer Science (FOCS 1979), 218-223.

\bibitem{ABRS}
B. Awerbuch, M. Betke, R. Rivest and M. Singh,
Piecemeal graph learning by a mobile robot,
Proc. 8th Conf. on Comput. Learning Theory (1995), 321-328.



\bibitem{BBFY}
E. Bar-Eli, P. Berman, A. Fiat and R. Yan,
On-line navigation in a room,
Journal of Algorithms 17 (1994), 319-341.

\bibitem{BFRSV}
M.A. Bender, A. Fernandez, D. Ron, A. Sahai and S. Vadhan,
The power of a pebble: Exploring and mapping directed graphs,
Proc. 30th Ann. Symp. on Theory of Computing (STOC 1998), 269-278.

\bibitem{BS}
M.A. Bender and D. Slonim, The power of team exploration:
Two robots can learn unlabeled directed graphs,
Proc. 35th Ann. Symp. on Foundations of Computer Science (FOCS 1994),
75-85.


\bibitem{BRS2}
M. Betke, R. Rivest and M. Singh,
Piecemeal learning of an unknown environment,
Machine Learning 18 (1995), 231-254.

\bibitem{BRS}
A. Blum, P. Raghavan and B. Schieber,
Navigating in unfamiliar geometric terrain,
SIAM Journal on Computing 26 (1997), 110-137.

\bibitem{BRT}
A. Borodin, W. Ruzzo, M. Tompa,
Lower bounds on the length of universal traversal sequences,
Journal of Computer and System Sciences 45 (1992), 180-203.



\bibitem{CDK}
J. Chalopin, S. Das, A. Kosowski,
Constructing a map of an anonymous graph: Applications of universal sequences,
Proc. 14th International Conference on Principles of Distributed Systems (OPODIS 2010), 119-134.





\bibitem{DKP}
X. Deng, T. Kameda and C. H. Papadimitriou,
How to learn an unknown environment I: the rectilinear case,
Journal of the ACM 45 (1998), 215-245.




\bibitem{DP}
D. Dereniowski, A. Pelc, Drawing maps with advice,  Journal of Parallel and Distributed Computing 72 (2012), 132--143.







\bibitem{DFKP}
K. Diks, P. Fraigniaud, E. Kranakis, and A. Pelc,
Tree exploration with little memory,
Journal of Algorithms 51 (2004), 38-63.


 \bibitem{DiPe}
 Y. Dieudonn\'{e}, A. Pelc, Impact of Knowledge on Election Time in Anonymous Networks. CoRR abs/1604.05023 (2016).
 
 \bibitem{DKM}
S. Dobrev, R. Kralovic, and E. Markou, Online graph exploration with advice,  Proc. 19th International Colloquium on Structural Information and Communication Complexity (SIROCCO 2012), 267-278.





\bibitem{DKK}
C.A. Duncan, S.G. Kobourov and V.S.A. Kumar,
Optimal constrained graph exploration,
Proc. 12th Ann. ACM-SIAM Symp. on Discrete Algorithms (SODA 2001), 807-814.

\bibitem{EFKR}
Y. Emek, P. Fraigniaud, A. Korman, A. Rosen, Online computation with advice, Theoretical Computer Science 412 (2011), 2642--2656.




\bibitem{FGIP}
P. Fraigniaud, C. Gavoille, D. Ilcinkas, A. Pelc,
Distributed computing with advice: Information sensitivity of graph coloring,
Distributed Computing 21 (2009), 395--403.

\bibitem{FIP1}
P. Fraigniaud, D. Ilcinkas, A. Pelc,
Communication algorithms with advice, Journal of  Computer and System Sciences 76 (2010), 222--232.

\bibitem{FIP2}
P. Fraigniaud, D. Ilcinkas, A. Pelc,
Tree exploration with advice, Information and Computation 206 (2008), 1276--1287.

\bibitem{FKL}
P. Fraigniaud, A. Korman, E. Lebhar,
Local MST computation with short advice,
Theory of Computing Systems 47 (2010), 920--933.


\bibitem{FI}
P. Fraigniaud and D. Ilcinkas.
Directed graphs exploration with little memory,
Proc. 21st Symposium on Theoretical Aspects of
Computer Science (STACS 2004), 246-257.

\bibitem{FP}
E. Fusco, A. Pelc, Trade-offs between the size of advice and broadcasting time in trees, Algorithmica 60 (2011), 719--734.


\bibitem{FPR}
E. Fusco, A. Pelc, R. Petreschi, Use knowledge to learn faster: Topology recognition with advice, Proc. 27th International Symposium on Distributed Computing (DISC 2013), 31-45.

\bibitem{GPPR02}
C.~Gavoille, D.~Peleg, S.~P\'{e}rennes, R.~Raz.
Distance labeling in graphs,
Journal of Algorithms 53 (2004), 85-112.

\bibitem{GMP}
C. Glacet, A. Miller, A. Pelc, Time vs. information tradeoffs for leader election in anonymous trees,
Proc. 27th Annual ACM-SIAM Symposium on Discrete Algorithms (SODA 2016), 600-609.

\bibitem{IKP}
D. Ilcinkas, D. Kowalski, A. Pelc,
Fast radio broadcasting with advice,
 Theoretical Computer Science, 411 (2012),  1544--1557.

 \bibitem{KKP05}
A. Korman, S. Kutten, D. Peleg, Proof labeling schemes,
Distributed Computing 22 (2010), 215--233.









\bibitem{SN}
N. Nisse, D. Soguet, Graph searching with advice,
Theoretical Computer Science 410 (2009), 1307--1318.




\bibitem{PaPe}
P. Panaite and A. Pelc,
Exploring unknown undirected graphs,
Journal of Algorithms 33 (1999), 281-295.

\bibitem{PaPe2}
P. Panaite, A. Pelc,
Optimal broadcasting in faulty trees,
Journal of Parallel and Distributed Computing 60 (2000), 566-584.

\bibitem{PeTi}
A. Pelc, A. Tiane, Efficient grid exploration with a stationary token,
International Journal of Foundations of Computer Science 25 (2014), 247-262.

\bibitem{RKSI}
N. S. V. Rao, S. Kareti, W. Shi and S.S. Iyengar,
Robot navigation in unknown terrains: Introductory survey of non-heuristic
algorithms,
Tech. Report ORNL/TM-12410, Oak Ridge National Laboratory, July 1993.

\bibitem{Re}
O. Reingold, Undirected connectivity in log-space, Journal of the ACM 55 (2008).























\end{thebibliography}
\end{document}